\documentclass[sigconf, nonacm]{acmart}
\usepackage{pvldb}

\renewcommand\vldbdoi{XX.XX/XXX.XX}
\renewcommand\vldbpages{XXX-XXX}
\renewcommand\vldbavailabilityurl{https://github.com/akhileshvb/accelerator-codegen/tree/artifact} 

\usepackage[inline]{enumitem}
\usepackage[dvipsnames]{xcolor}
\usepackage{todonotes}
\usepackage{lipsum}
\usepackage{subcaption}
\usepackage{cleveref}
\usepackage{makecell}
\usepackage{graphicx}
\usepackage{tabularx}
\usepackage{bm}
\usepackage{pifont}

\newcommand{\sparagraph}[1]{\vspace{1mm}\noindent {\bf #1}}
\newcommand{\thesystem}{JetStream}
\newcommand{\qsem}{QSem}

\usepackage{adjustbox}
\usepackage{array}
\usepackage{subcaption}
\newcolumntype{R}[2]{%
    >{\adjustbox{angle=#1,lap=\width-(#2)}\bgroup}%
    l%
    <{\egroup}%
}

\newcommand{\Sup}{\textcolor{ForestGreen}{\textbf{\ding{51}}}}
\newcommand{\NSup}{\textcolor{red}{\textbf{\ding{55}}}}
\newcolumntype{Y}{>{\centering\arraybackslash}X}

\begin{document}

\title{JetStream: Generating Verified Query Accelerators for Existing Database Systems}

\author{Akhilesh Balasingam}
\orcid{0009-0003-8637-896X}
\affiliation{%
  \institution{MIT CSAIL}
  \city{Cambridge}
  \state{MA}
  \country{USA}
}
\email{akhilvb@mit.edu}

\author{Amadou Ngom}
\orcid{0009-0000-8422-5212}
\affiliation{%
  \institution{MIT CSAIL}
  \city{Cambridge}
  \state{MA}
  \country{USA}
}
\email{ngom@mit.edu}

\author{Geoffrey X. Yu}
\orcid{0009-0005-3186-1465}
\affiliation{%
  \institution{MIT CSAIL}
  \city{Cambridge}
  \state{MA}
  \country{USA}
}
\email{geoffxy@mit.edu}

\author{Tim Kraska}
\orcid{0009-0003-2414-2759}
\affiliation{%
  \institution{MIT CSAIL}
  \city{Cambridge}
  \state{MA}
  \country{USA}
}
\email{kraska@mit.edu}

\begin{abstract}
% TODO: clunky
Recent work has shown that LLMs can synthesize highly instance-optimized database systems for fixed workloads, but to date these systems only replace the DBMS's native storage and assume static data: users must migrate their data, often give up ad-hoc queries, lose acceleration once the data changes, and lack formal correctness guarantees. We present JetStream, which instead generates query-specific accelerators that extend an existing DBMS. Given a set of parameterized query templates, a staged agentic workflow synthesizes accelerators that run directly over the DBMS's native tables and may maintain custom auxiliary state, while the host DBMS remains responsible for storage, transactions, recovery, and all the other queries. A fixed, engine-neutral substrate keeps auxiliary state transactionally consistent under writes, while generated backend adapters handle engine-specific hooks. JetStream generates two kinds of accelerators: tested accelerators, which may use any optimization and are validated by differential testing, and verified accelerators, which trade some performance for a machine-checked proof that every answer satisfies a trusted specification derived from the SQL. On TPC-H at SF=20, JetStream achieves a $584\times$ geomean read speedup over DuckDB ($89\times$ verified) and retains $139\times$ ($46\times$ verified) under a continuous update stream of 10,000 modified rows per second.

\end{abstract}

\maketitle

%%% VLDB block start %%%
\vldbtopmatter
%%% VLDB block end %%%

\section{Introduction}\label{sec:intro}
% \singlefigplaceholder{2.5in}{Motivating result: Workload speedup on our system with
% comparison to GenDB, BespokeOLAP, Base database system (probably DuckDB). Construct
% a workload with sufficient writes. Plot speedup, auxiliary
% space usage. Result should show Ours > \{GenDB, BespokeOLAP\} > Base system. It
% might also be that Ours > Base > \{GenDB, BespokeOLAP\} with sufficient writes.
% We should include both this realistic setting (i.e., with writes) and a
% read-only setting (the scenario GenDB/BespokeOLAP were designed for).}

% TODO:
% move to bar chart, move mem footprint to later section
% add BRAD in 1 sentence

\begin{figure}[t]
    \centering
    \includegraphics[width=\columnwidth]{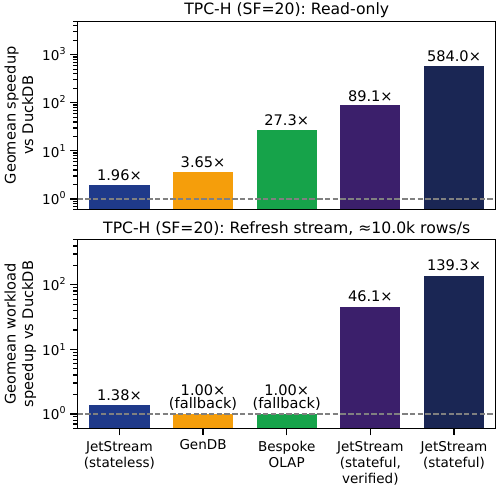}
    \caption[Geomean speedup over DuckDB on TPC-H SF=20]{%
      Geomean speedup over DuckDB on TPC-H SF=20, read-only (top) and under the standard TPC-H refresh stream ($\approx$10K modified rows/s; bottom). All systems, including DuckDB, use 40 threads, except verified accelerators.
      Stateful \thesystem{} reaches $584\times$ read-only ($89.1\times$ verified) and keeps $139.3\times$ under writes ($46.1\times$ verified). GenDB and Bespoke OLAP do not maintain state under writes and fall back to DuckDB
      ($1\times$).\protect\footnotemark}
    \vspace{-1em}
    \label{fig:headline}
\end{figure}
\footnotetext{Numbers differ from the arXiv version, which used single-threaded implementation and a write rate of 2.5k rows/sec.}

\newcommand{\Yes}[1]{%
  \makecell[l]{\makebox[\linewidth][c]{\large\Sup{}}\\[0.1em] #1}}
\newcommand{\No}{\makecell[c]{\large\NSup{}}}

\newcommand{\Partial}{\textcolor{orange!85!black}{\large$\boldsymbol{\circ}$}}

\begin{table}[t]
\caption{No existing query accelerator synthesis system supports all desired practical capabilities.}
\label{tbl:qual-comparison}
\vspace{-0.5em}
\small
\setlength{\tabcolsep}{3pt}
\begin{tabularx}{\columnwidth}{@{}l>{\raggedright\arraybackslash}X cccc@{}}
  \toprule
  \textbf{System} & \textbf{Approach} &
    \makecell{\textbf{Ad-hoc}\\\textbf{queries}} &
    \makecell{\textbf{Accel.}\\\textbf{under}\\\textbf{writes}} &
    \makecell{\textbf{Native}\\\textbf{DBMS}\\\textbf{storage}} &
    \makecell{\textbf{Proven}\\\textbf{correct}} \\
  \midrule
  Castor~\cite{castor-feser20} & Deductive synthesis & \NSup{} & \NSup{} & \NSup{} & \Partial{}$^{a}$ \\
  GenDB~\cite{gendb} & Agentic synthesis & \NSup{} & \NSup{} & \NSup{} & \NSup{} \\
  Bespoke OLAP~\cite{bespoke-olap} & Agentic synthesis & \Sup{}$^{b}$ & \NSup{} & \NSup{} & \NSup{} \\
  \midrule
  \textbf{\thesystem{}} & Agentic synthesis w/ gen. maint. & \Sup{} & \Sup{} & \Sup{} & \Sup{}$^{c}$ \\
  \bottomrule
  \multicolumn{6}{@{}p{\columnwidth}@{}}{\footnotesize\rule{0pt}{2.2ex}$^{a}$Correct by construction through semantics-preserving rewrites, without a machine-checked proof for each generated program. $^{b}$Through a separate DBMS that holds a second copy of the data. $^{c}$Optional: verified accelerators carry a machine-checked proof against a specification derived from the SQL; tested accelerators use differential testing.}
\end{tabularx}
\vspace{-1em}
\end{table}

% Story:
% - Relational DBMSes designed to be general
% - If you know which queries you expect to run over and over, you can clearly
%   run it faster with a specialized system
% - Up until now, you would only build such a system if the workload is
%   particularly important (e.g., revenue generating)
% - Now with LLMs, this calculation has changed
%
% - Recent work has explored using LLMs to synthesize query engines for specific
%   query templates
% - Describe GenDB, BespokeOLAP
% - Both systems good for static databases
% - However, we argue that both systems lack capabilities needed for practicality
% - Supporting unseen queries, preserving perf. benefits under writes, and not duplicating the storage
% - Achieving the properties above is challenging because [...]

% - Introduce JetStream, first by highlighting the performance results
% - Then point back to Table 1 as we go
% - Describe correctness guarantees -> same as GenDB, Bespoke OLAP - empirical
%   testing, which we argue is no worse than existing systems that also rely on
%   testing

% - Give summary of key results, contributions list
% This generality, however, leaves performance on the table.

Relational database systems are general-purpose: they support diverse schemas and data and can execute any query expressible in their SQL dialect. This generality costs performance. An aggregation, for example, is generally implemented with a hash table, but if the set of groups is known ahead of time, a direct-mapped array eliminates hash probing altogether. Knowledge of a query workload and its data enables such specialized execution strategies, which we call \emph{accelerators}. Historically, building accelerators required substantial engineering effort and sacrificed generality, so they were built only for performance- or revenue-critical workloads~\cite{scuba-abraham13, napa-agiwal21, f1lightning-yang20}. Recent advances in LLM code generation~\cite{openai2025codex, anthropic2025claudecode, wang2024openhands} change this trade-off by making it practical to construct accelerators for known workloads automatically~\cite{gendb, gptdb-trummer23, bespoke-olap}. This raises a natural question: can we obtain the benefits of specialization without giving up the generality and practicality of a conventional database system?

GenDB~\cite{gendb} and Bespoke OLAP~\cite{bespoke-olap} take a step in this direction by synthesizing query engines that specialize the storage layout, data structures, and query implementations to a known workload. However, their design decisions limit their practicality (\Cref{tbl:qual-comparison}). First, they rewrite the base data into a layout chosen for the workload, a re-encoded copy that acts like a simple materialized view of the tables it reads, which forces users to migrate it, complicates writes, and may require rewriting it again when the workload changes. Second, their data structures assume static data and lack incremental maintenance, so after writes they either return stale results or lose their benefits. Third, GenDB does not support ad-hoc queries at all. Fourth, both validate correctness only by testing on sampled parameters and data, which cannot rule out errors on inputs outside the tested distribution. A practical system should instead retain the generality of a conventional DBMS, supporting writes and ad-hoc queries over its native storage, while accelerating known, pre-declared query templates.

Three challenges make this difficult. First, \emph{reward hacking}~\cite{metr-2025-recent-reward-hacking}: coding agents can modify the system and invoke their own tests and benchmarks, so apparent improvements may come from weakened correctness checks or measurement artifacts. Second, \emph{engine integration}: engines expose different scan APIs, mutation paths, and transaction mechanisms, and encoding these details into every accelerator would couple it to one engine and force synthesis to rediscover the same integration logic repeatedly. Third, \emph{verification}: formal verification offers the strongest correctness guarantee, but the optimizations that make accelerators fast also make them harder to prove correct.

We present \textbf{\thesystem{}}, which synthesizes query-specific accelerators and integrates them within an existing database system. For each query template, \thesystem{} generates an accelerator that runs directly over the DBMS's native storage, and may maintain auxiliary state. This auxiliary state acts as a far more capable materialized view. Rather than a relation the engine must still scan and aggregate, it uses data structures chosen for the query, such as prefix sums over the parameters' domains, and generated code maintains it incrementally instead of through the more costly synchronous view maintenance. The DBMS, then, remains responsible for storage, transactions, and recovery. Queries that do not have accelerators will run as before, and accelerators can be added and removed without regenerating the entire system. \thesystem{} builds on Tailwind \cite{tailwind-yu26}, which integrates hand-written accelerators into existing systems. Unlike Tailwind, \thesystem{} instead generates the accelerators and their maintenance logic, addressing each of the challenges mentioned above. To prevent reward hacking, an external orchestrator owns every test and benchmark. For engine integration, accelerators are all written against a fixed, engine-neutral substrate, and \thesystem{} generates one backend adapter per engine to handle its scans, write capture, and transactions. \thesystem{} offers two kinds of stateful accelerators. \textbf{Tested} accelerators may use any optimization and are validated by differential testing against the host, which covers only the inputs tested with a wide range of read and write workloads. To address the challenges of verification, \textbf{verified} accelerators are written in a constrained programming model, and are proved against a trusted, frozen specification generated deterministically from the SQL.    

As shown in Figure~\ref{fig:headline}, \thesystem{} provides both aggressive specialization and practical update support. On TPC-H at SF=20, stateful \thesystem{} achieves a $584\times$ geomean read-only speedup over DuckDB, compared with $27\times$ for Bespoke OLAP and $3.7\times$ for GenDB, and verified accelerators reach $89\times$. Under a TPC-H refresh stream of about 10K modified rows per second, \thesystem{} retains a $139\times$ geomean workload speedup ($46\times$ verified), while systems without write maintenance can at best fall back to the performance of the base DBMS. On SEC-EDGAR~\cite{gendb}, a workload built from real financial filings and unlikely to appear in LLM training data, \thesystem{} achieves $225\times$ and $5335\times$ geomean read-only speedups over DuckDB and PostgreSQL, respectively. And, against extensively tuned materialized views maintained by pg\_ivm in PostgreSQL, \thesystem{} sustains $24\times$ the write rate while running up to $28.5\times$ faster and using 39\% less memory.

In summary, we make the following contributions:
\begin{itemize}[leftmargin=*]
\item \textbf{A profiling-driven agentic workflow for generating accelerators that extend an existing DBMS.} It separates accelerator generation from evaluation and produces highly optimized accelerators that run over the DBMS's native storage while it continues to serve other ad-hoc queries.
\item \textbf{An engine-neutral substrate and a framework for write-aware maintenance of auxiliary state.} The substrate decouples accelerators from engine-specific integration, and the framework incrementally updates, rebuilds, or invalidates auxiliary state to keep it transactionally consistent with the base data under writes.
\item \textbf{A methodology for generating verified accelerators.} \thesystem{} generates accelerators with machine-checked proofs that every answer they return satisfies a trusted specification derived from the SQL, retaining substantial speedups while providing a stronger correctness guarantee than differential testing.
\end{itemize}

\section{System Overview}
\label{sec:system-architecture}

\begin{figure*}[t]
    \centering
    \includegraphics[width=0.98\textwidth]{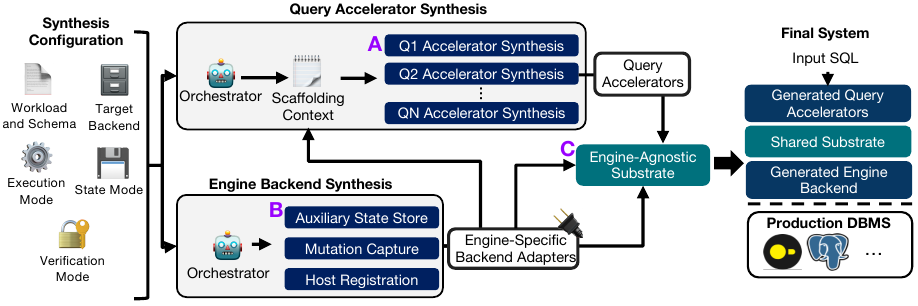}
    \caption{\thesystem{} system overview. Given a synthesis configuration, \thesystem{} synthesizes query accelerators (A) and engine-specific backend adapters (B), both against a fixed engine-agnostic substrate (C). The resulting system can then execute accelerated queries inside an unmodified production DBMS.
    (Section~\ref{sec:system-architecture}).}
    \label{fig:system}
\end{figure*}

\thesystem{} accelerates a known set of queries while leaving the underlying DBMS in charge of everything else. Figure \ref{fig:system} shows how it constructs such a system. 

\sparagraph{Inputs.} A user provides a synthesis configuration consisting of the workload (a schema and a set of parameterized query templates), the target DBMS, and three modes: the \emph{execution mode} (single- or multi-threaded), the \emph{state mode} (stateless or stateful), and the \emph{verification mode} (differentially tested or verified).

\sparagraph{Synthesis.} \thesystem{} runs two agentic pipelines. The first generates an accelerator for each template (\Cref{fig:system}A, \Cref{sec:workflow}). The second generates a backend for the target DBMS (\Cref{fig:system}B, \Cref{sec:backend-synth}): an auxiliary state store, mutation capture, and host registration. Both are built against the fixed, engine-agnostic substrate (\Cref{fig:system}C, \Cref{sec:substrate-maintenance}), which is the only interface between the query-specific code and the engine-specific code. 

% Orchestrators control the execution of both pipelines. In the query accelerator synthesis pipeline, the orchestrator assembles scaffolding context, which describes the substrate interface, the runtime backend adapters, the workload and configuration, and then synthesizes an accelerator for each template concurrently (\ref{sec:workflow}). In the engine backend synthesis pipeline, the orchestrator implements the substrate's engine-specific services: an auxiliary state store for accelerator state, mutation capture, which observes the rows that committed transactions, and host registration, which hooks accelerators into the DBMS so that matching queries are routed to them. Because accelerators depend only on the substrate, the same accelerators can be deployed on a different DBMS by synthesizing a new backend. 
% \todo{idk if this is realistic or necessary}

\sparagraph{Stateless and stateful accelerators.} Each accelerator is a table-valued function whose arguments are the template's parameters. Stateless accelerators compute each result from base-table scans, specializing only transient data structures and execution strategy. Stateful accelerators also keep the auxiliary state described in \Cref{sec:intro}, most commonly partial aggregates organized by the template's parameters, such as interval trees over a date predicate, so a query only combines a few precomputed entries at run time.

\sparagraph{Maintenance under writes.} For a stateful accelerator to stay correct, its auxiliary state must reflect committed base data. The backend captures each transaction's writes, and the substrate routes them to every affected accelerator, whose generated maintenance logic updates its state incrementally or rebuilds it, as a cost model decides
(\Cref{sec:cost-model}). State is versioned, so each query only reads the
version that matches its snapshot of the base data.

% \todo{verus citation}
\sparagraph{Tested and verified accelerators.}
Tested accelerators are validated by differential testing against the
host, which covers only the inputs it samples. Verified accelerators
instead carry a machine-checked proof, written in Verus~\cite{verus},
that every answer meets a specification derived from the SQL, under
explicit contracts on the host (\Cref{sec:verification}); proofs favor
simpler designs, so users trade speed for a stronger guarantee.

\section{Agentic Pipelines}
\label{sec:workflow}
% \begin{itemize}[leftmargin=*]
%   \item Discuss agentic workflow
%   \item Discuss system architecture
%   \item Describe key idea for keeping auxiliary state up-to-date
%   \item Describe correctness assumptions
%   \item Probably: more detailed discussion about handling writes to the
%     auxiliary state
% \end{itemize}

% \noindent
% \gy{The system overview and details section will likely change. We should
% probably just start by writing out everything that we deem interesting and then
% we can work on distilling/reorganizing it into sections/subsections for the
% paper.}

Accelerator synthesis is a multi-stage design problem with delayed, noisy feedback. Early choices of execution strategy and state constrain the later optimizations, but their quality only shows once a concrete accelerator is measured. As we later show in \Cref{sec:eval}, a single agent loop handles this poorly. It struggles to separate design from implementation problems \cite{cai2026forkandflushescapingideabasins}, and, in harder cases, edits the tests and DBMS code it is judged by. \thesystem{} instead structures accelerator generation as the workflow shown in Figure \ref{fig:agent-synthesis}, which separates high-level design, implementation, optimization, and audit into distinct stages. This section describes the agent pipeline for generating \textit{tested} accelerators and the engine backends.

\begin{figure*}[t]
\centering
\includegraphics[width=\textwidth]{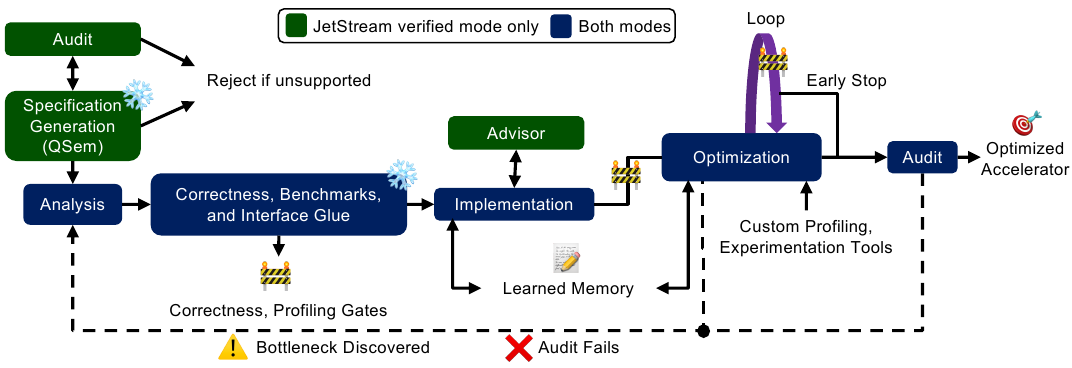}
\caption{\textbf{\thesystem{}'s accelerator synthesis pipeline.} Each box is a stage run by a separate agent. Blue stages run for all accelerators, and green stages only for verified ones (\Cref{sec:verification}). Frozen artifacts cannot be changed by later agents, and gates mark the checks a candidate must pass to advance.}
\label{fig:agent-synthesis}
\end{figure*}

\sparagraph{Staged Accelerator Synthesis.}
\label{sec:staged-synthesis}
Each stage in Figure \ref{fig:agent-synthesis} is executed by a separate agent instance with stage-specific context, instructions, and tools. The analysis stage proposes several materially different designs. Each specifies an execution strategy and, for stateful accelerators, the auxiliary state and how to maintain it. \thesystem{} selects the candidate with the greatest expected read-path benefit whose maintenance cost grows at most linearly with scale factor, and an implementation agent builds it against the substrate (\Cref{sec:substrate-maintenance}). Before optimization begins, the orchestrator constructs the differential tests and benchmarks and keeps ownership of them, so later stages cannot weaken correctness checks or tune the benchmark to favor a candidate. After optimization converges, a separate audit agent tests beyond what earlier stages saw. If the audit fails, or optimization finds that the design has reached a structural limit, the pipeline returns to analysis; otherwise, the accelerator is finalized.

\sparagraph{Evidence-Guided Optimization.}
\label{sec:evidence-optimization}
Each optimization attempt is judged by measurement: the orchestrator
builds the candidate, runs the correctness gates, and A/B benchmarks it
against the best validated version, which it replaces only if it is
correct and improves the objective. When aggregate latency hides the
bottleneck, the agent can add phase timers and query-specific counters
to separate the costs of base scans, state lookups, decoding,
synchronization, and maintenance. For stateful accelerators, the
objective covers both reads and writes, and the orchestrator varies the
write frequency across attempts so the agent does not overfit to one
trade-off. Between attempts, \thesystem{} compresses failed hypotheses,
bottlenecks, successful changes, and useful profiles into a learned
memory, so later agents avoid repeating work without carrying the full
history. Optimization stops early once attempts stop improving.

\sparagraph{Concurrent Workload Synthesis.}
\label{sec:concurrent-synthesis}
% - worker private state
% - explicit shared state
% - ownership-aware recovery
% - token-cost cutting measures like repo context preloaded, etc.
Each query's accelerator is synthesized by its own asynchronous worker,
which advances through the stages independently. All accelerators
compile into one shared library, so each worker has a private worktree,
build directory, and agent runtime, and edits to shared registration
files go through a lease. When a failure comes from shared or peer
files, the orchestrator restores them from the last validated snapshot
rather than having the agent debug foreign code. Workers also receive a
compact summary of the integration points and existing accelerators, so
they spend fewer turns exploring the repository.

\sparagraph{Engine Backend Synthesis.}
\label{sec:backend-synth}
% \begin{figure}[t]
%     \centering
%     \includegraphics[width=\columnwidth]{system_overview/backend_synth.pdf}
%     \caption{\thesystem{}'s engine-specific adapter synthesis.}
%     \label{fig:backend-synthesis}
% \end{figure}
Every accelerator for an engine relies on the same integration logic: base-table scans, write capture, transaction hooks, and isolation of uncommitted changes. This is also where errors are most dangerous, since a missed write or a visibility bug silently yields stale results. \thesystem{} therefore generates this logic once per engine, as a backend behind the substrate interface (Figure~\ref{fig:system}B). Because every accelerator on that engine shares the backend, it can be validated more thoroughly than any single accelerator, per-query agents stay confined to query logic, and \thesystem{} can target a user's existing DBMS without handwritten integration.

\section{Substrate and Auxiliary State Maintenance}
\label{sec:substrate-maintenance}
% - engine-agnostic substrate
% - engine adapter generation
% - stateless over native DBMS storage, stateful over new interface

% \begin{figure}[t]
%     \centering
%     \includegraphics[width=\columnwidth]{system_overview/shared_substrate.pdf}
%     \caption{The fixed substrate routes captured writes to affected
%     accelerators, plans and publishes maintenance at commit, and serves
%     reads from versioned auxiliary state.}
%     \label{fig:substrate}
% \end{figure}

The substrate is the fixed layer between generated accelerators and the generated engine backend (Figure \ref{fig:system}C). It keeps each accelerator's auxiliary state consistent with committed base data using standard techniques, so we describe only what the accelerators and the cost model rely on.

\subsection{Capturing and Publishing Maintenance}
\label{sec:maintenance}

% \sparagraph{Accelerator interface.}
As discussed in \Cref{sec:system-architecture}, each stateful accelerator is built against the engine-neutral substrate. Each stateful accelerator declares the tables and columns it reads, builds its state from base tables, and translates captured writes into query-specific changes (e.g., signed contributions to affected groups). For a chosen maintenance action, it produces the next version of its state, and it provides a fallback that computes its result directly over base tables. To localize maintenance, we allow \thesystem{} to divide its state into logical \textit{partitions} by a key computed from row values. One example of this from TPC-H Q4. It keeps order counts per priority, and partitions them by order month. An update that moves an order from March to April affects exactly only those two partitions.

\sparagraph{Write capture and commit.}
The backend translates each engine's write and transaction events into a common format, which includes the table, the operation, and the old and new values of the columns the accelerators read. Before the engine commits a transaction, the substrate reserves a place in the commit order for every accelerator that reads a modified table. After the commit, every affected accelerator computes its next state and publishes it in commit order. Maintenance itself does not change the outcome of a transaction. If the transaction aborts, its captured changes are discarded and the reserved place for it in the commit order is released.

\sparagraph{Versioned state.}
\thesystem{} follows standard multi-version concurrency control \cite{bernstein-mvcc} for maintaining auxiliary state: each published state is an immutable version. This means that readers will not observe a partial state. Versions are tagged with the commit epoch from the host engine, so each query reads the newest version visible to its snapshot. This means that auxiliary state becomes visible to the transaction exactly when the corresponding base table changes do. The exception is a transaction's own uncommitted writes. Since no published version reflects them, a call from a transaction that has written to a table the accelerator reads uses the fallback over base tables. A reader with a snapshot that includes a pending commit waits for it to be published. Readers with older snapshots never wait, since the version they need is already published. Because each accelerator has its own commit lock, a commit blocks only readers of the accelerators it affects.

\subsection{Maintenance Actions and Cost Model}
\label{sec:cost-model}

When a transaction changes rows an accelerator reads, a planner chooses how to bring its state up to date. It can \textit{apply the delta}, updating the state incrementally from the transaction's changes; \textit{rebuild the affected partitions} from the base tables; \textit{rebuild everything}; or \textit{defer}, marking the state invalid so that the next reader triggers a rebuild. The planner picks the action with the lowest estimated cost, counting on both the work done at commit and the work pushed onto reads before the next write. It estimates these costs (in wall clock time) with simple linear models whose coefficients are calibrated offline per accelerator and engine. They are summarized:
\begin{align*}
C_{\mathrm{delta}} &= c_i + c_d + c_r R,
  & C_{\mathrm{all}} &= c_i + c_f,\\
C_{\mathrm{part}} &= c_i + c_{p_0} + c_p P,
  & C_{\mathrm{defer}} &= c_i + \min(1, Q)\, c_f,
\end{align*}
where $R$ is the number of rows the transaction changed, $P$ the number of affected partitions, $c_i$ the cost of installing a new version, $c_f$ the cost of a full rebuild, and $Q$ the expected number of reads before the next write. \thesystem{} estimates $Q$ online as the average number of reads between consecutive writes over a sliding window. Deferring pays $c_f$
only if a read arrives before the next write, so it wins only when
writes outnumber reads ($Q < 1$).
% Costs within 1 ms of each other count as ties, which favor immediate actions.

\section{Verified Accelerators}
\label{sec:verification}

Differential testing shows only that an accelerator is correct on the inputs that were sampled. \thesystem{} can also generate \textit{verified} accelerators, which carry a machine-checked proof that they compute what their SQL query means (according to a trusted translator), relative to explicit contracts on the host DBMS and interface that connects it to the accelerator. Informally, if the DBMS supplies correct inputs and applies its commits in order, then for every parameter value and every committed database state, any answer a verified accelerator returns is correct (Theorem \ref{thm:verif}). 
% A second machine-checked guarantee limits when an accelerator may decline a call and leave the query to the host DBMS: on any database the DBMS can store, it declines only when no correct answer exists (the SQL itself has no result, e.g., a zero divisor) or none fit the result type (e.g., integer overflow). \Cref{sec:verif-guarantee}
% lists these cases and lists everything that the guarantees trust. 
Every step of this argument is machine-checked, relative to the input and commit-order contracts, which appear in Verus as trusted preconditions. \Cref{tab:verif-chain} summarizes what is checked and what is trusted. Verified accelerators currently run single-threaded within execution. 
% \Cref{tab:verif-chain} provides a walk-through of what is proved and assumed in our verification model.

\begin{table}[t]
\setlength{\tabcolsep}{3pt}
\renewcommand{\arraystretch}{0.95}
\begin{tabularx}{\columnwidth}{@{}>{\raggedright\arraybackslash}X l l@{}}
\toprule
\textbf{Component} & \textbf{Source} & \textbf{Status} \\
\midrule
Spec $\mathsf{Q}$, input contracts, $\alpha$ & \qsem{} & Trusted \\
Spec sanity lemmas (\S\ref{sec:verif-spec}) & \qsem{} & Checked \\
\texttt{build}/\texttt{apply}/\texttt{execute} proofs & Agent & Checked \\
Writes keep state correct (L\ref{lem:seq}) & Authors & Checked \\
Reader gets right version (L\ref{lem:store}) & Authors & Checked \\
Version matches snapshot (L\ref{lem:snap}) & Authors & Checked \\
Answers are correct (T\ref{thm:verif}) & Agent & Checked \\
No spurious declines (T\ref{thm:total}) & Agent & Checked$^{*}$ \\
Host contract, $\tau$, engine model & Host & Assumed \\
Glue, scans, change capture & Agent & Assumed \\
Output formatting, fallback & Host & Trusted \\
DataFusion, Verus, SMT, \texttt{rustc} & -- & Trusted \\
\bottomrule
\end{tabularx}
\caption{What the verified guarantee rests on (L/T: Lemma/Theorem).
\emph{Checked}: by Verus. \emph{Assumed}: required of the host, stated as
trusted preconditions, tested but not proved. \emph{Trusted}: taken as
correct. $^{*}$Given that every storable database lies in the envelope.}
\label{tab:verif-chain}
\vspace{-2em}
\end{table}

\sparagraph{Approach.}
Our machine-checked proofs have three parts:
\begin{enumerate}[leftmargin=*]
\item a formal statement of what the query means (Section~\ref{sec:verif-spec});
\item a contract that describes the data the DBMS gives the accelerator and what each of the accelerator's operations must establish (Section~\ref{sec:verif-boundary}), extended to updates and concurrent readers (Section~\ref{sec:verif-stateful});
\item a proof that the accelerator's code meets that contract.
\end{enumerate}

\qsem{}, a translator we built, generates the first two mechanically from
the SQL, and we wrote and verified the versioned store once,
independently of any workload. The agentic pipeline writes only the third part and cannot modify (1), (2) or the store, so an agent can change \textit{how} something is proven
but not \textit{what}. This means that a flawed or incorrect accelerator fails verification rather than returning wrong answers. We write accelerators and proofs in Verus~\cite{verus}, a verifier for Rust whose annotations are checked at compile time by an SMT solver and
erased before execution, so proofs themselves add no runtime cost.

\sparagraph{Generation Pipeline.}
Verified accelerators are generated by the same staged pipeline as tested ones, with additional stages, as shown in \Cref{fig:agent-synthesis}. Everything that defines what is proven, or that is trusted, is written and frozen before any agent writes the code that must satisfy it. \qsem{} generates the query specification; if the query template is unsupported, then the pipeline does not continue. As with the {tested} accelerators, an analysis agent designs the accelerator's state and inputs, which a critic agent reviews. From this design, \qsem{} generates the external-facing contract: what the accelerator may assume about the data the host DBMS supplies. Three agents then work in parallel to audit that contract, write benchmarking suites, and write the glue that connects the host DBMS to the to-be-implemented Verus kernel. Finally, in one session, the agent writes the accelerator, tests it, and proves it. A candidate is only accepted if it passes Verus verification, the frozen files are unchanged, and the code contains no proof-weakening construct like \texttt{assume}. Stalling proofs are aided by a more powerful advisor agent that can prevent the proof agent from falling into unproductive idea sinks. Agents can also contribute general lemmas to a shared proof library, helping accelerate proof-writing for other accelerators.

\subsection{Specifying What a Query Means}
\label{sec:verif-spec}

A query's meaning must be stated over a model of the database, and we will provide a high-level overview of that model here. For each table $T$ that a query reads, \qsem{} defines a row type $R_T$ that contains only the columns the query reads. A model database $D$ assigns each table a finite bag $D.T$ of $R_T$ rows. Values are mathematical integers, byte strings, booleans, and doubles, with decimals and dates encoded as integers.
Integer and decimal arithmetic is exact, and rounding to the engine's result type is part of output formatting; \texttt{SUM} and \texttt{AVG} over doubles are specified up to a proven rounding bound. Text compares bytewise, and engine-specific behavior, such as \texttt{NULL} ordering and division by zero, comes from an engine profile. This model is part of the trusted base (\Cref{sec:verif-guarantee}).

\qsem{} uses DataFusion~\cite{datafusion} to parse the SQL and to resolve names and types. It takes the unoptimized logical plan, which keeps the query's original shape, and translates the plan clause by clause into Verus definitions. Figure~\ref{fig:q3-spec} shows the result for an example query, TPC-H Q3. The \texttt{FROM} and \texttt{WHERE} clauses become $T$, the bag of tuples in the cross product of the tables that satisfy the predicate. \texttt{GROUP BY} partitions $T$ into groups $G(k)$, and the full result contains one row for each non-empty group, with its aggregates computed over that group. \texttt{ORDER BY} requires that no row sorts strictly before a row that precedes it, and \texttt{LIMIT} takes a prefix. \textbf{The top-level definition is a relation
\(
  \mathsf{Q}(D, p, o),
\)
which holds exactly when the output sequence $o$ is a correct answer to the query on database $D$ with parameter values $p$.} $\mathsf{Q}$ here is a relation and not a function because SQL itself leaves some freedom: rows that tie under \texttt{ORDER BY} may appear in any order, and under \texttt{LIMIT}, which of the tied rows make the cut is not determined \cite{spja}.

\begin{figure*}[t]
    \centering
    \includegraphics[width=\textwidth]{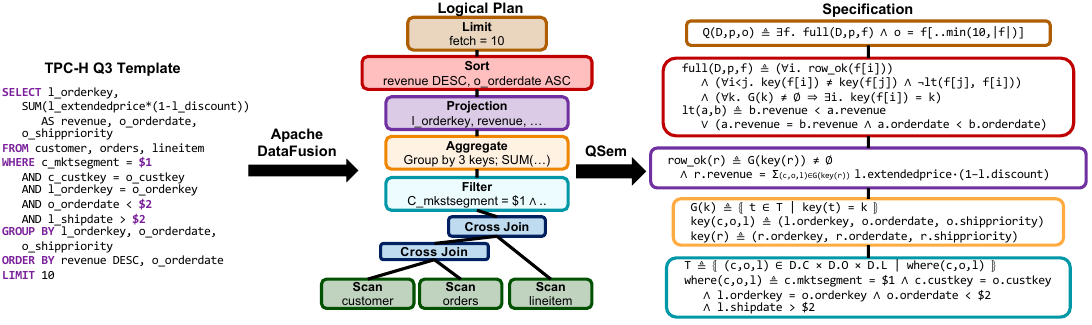}
    \caption{TPC-H Q3: the SQL template, the unoptimized logical plan that DataFusion produces, and the specification \qsem{} derives from the plan, one definition per plan node (shown in mathematical notation). $\mathsf{Q}(D, p, o)$ holds when $o$ is a correct answer on model database $D$ with parameters $p = (\$1, \$2)$. $D.C$, $D.O$, and $D.L$ are the \texttt{customer}, \texttt{orders}, and \texttt{lineitem} rows of $D$. Braces denote multisets, and $f[..n]$ is the first $n$ rows of $f$.}
    \label{fig:q3-spec}
\end{figure*}

Because $\mathsf{Q}$ is trusted, an error in its translation would silently change what every proof establishes. Some such errors, however, break properties that any correct query relation must have, and these can be checked mechanically. With each specification, \qsem{} therefore emits lemmas from fixed templates chosen by the query's structure. Verus checks them, so they add confidence in the translation without adding to the trusted base:
% \todo{more clear for most queries defined is trivial}
\begin{itemize}[leftmargin=*]
\item \emph{Existence}: $\mathsf{Q}$ has an answer on $(D, p)$ exactly when $\mathit{defined}(D, p)$ holds, where $\mathit{defined}$ characterizes the inputs on which SQL returns a result instead of raising a runtime error. Because \qsem{} refuses every other construct that can fail at runtime, such as scalar subqueries that may return several rows, the only runtime error it models is division by zero, on engines whose profile makes it an error.

\item \emph{Determinacy}: the full answer before \texttt{LIMIT} is unique as a bag, and it is unique as a sequence when \texttt{ORDER BY} orders every pair of output rows. This lemma rules out a specification that is too loose, for example, one that accepts any subset of the groups. QSem omits this lemma for queries with floating-point aggregates, whose rounding bound deliberately admits more than one correct answer.
\end{itemize}

\qsem{} supports SPJA queries~\cite{spja} with correlated subqueries,
CTEs, \texttt{DISTINCT}, \texttt{UNION}, single-table outer joins, and
top-level \texttt{ORDER BY}/\texttt{LIMIT}.
It handles \texttt{NULL}s wherever two-valued logic agrees with SQL's
three-valued logic, and refuses everything else, such as \texttt{NOT}
over a possibly-\texttt{NULL} comparison or window functions, before any
accelerator code exists.
The specification, obligations, and glue are then hashed, and every
orchestrator-run verification checks the hashes to make sure the agentic pipeline did not modify them when writing the accelerator.

\subsection{The Accelerator Contract}
\label{sec:verif-boundary}

A verified accelerator is implemented as a Verus crate that defines a state and three operations: \texttt{build} constructs a state from base-table data, \texttt{apply} updates a state with one transaction's changes, and \texttt{execute} answers the query from a state. Because the auxiliary state cannot reconstruct the database it is meant to summarize, each auxiliary state also needs to carry a ghost copy of the model database it represents. In this section, we will refer to this as $s.\texttt{db}()$ for the state $s$. An agent-written representation invariant $\mathit{Inv}(s)$ relates the state's data structures (e.g., sorted indexes) to $s.\texttt{db}()$. Because $\mathit{Inv}$ is declared as a Verus type invariant, Verus proves it wherever a state is constructed or modified, and \texttt{execute} can rely on it without a precondition.

\sparagraph{Inputs (trusted).}
Verified code never reads tables directly. Instead, the analysis agent declares each accelerator's inputs as SQL statements, and \qsem{} turns each declaration into an input contract with a fixed meaning that the host must satisfy. A different agent could choose different inputs for the same query, but a poor choice can only make the proof fail, since the accelerator must still prove the fixed query specification ( \Cref{sec:verif-spec}) with exactly what it receives. There are three kinds of input:
\begin{itemize}[leftmargin=*]
\item A \emph{scan} projects one base table, optionally filtered, with no parameters, joins, or aggregation. The rows passed to \texttt{build} equal the scan's result on database $D$ as a bag.
\item A \emph{delta} gives one transaction's changes to a table as signed rows: $+1$ for an insertion and $-1$ for a deletion. An update is a deletion plus an insertion. The full changed rows take $D$ to $D'$ as bags. The actual delta passed at runtime can be a partial row, like its key. In this case, the full rows are just a ghost argument. 
\item A \emph{rescan} lets the accelerator fetch rows near a change instead of storing every column. The accelerator computes a set of keys from its state and the deltas and the host returns the rows of  $D'$ with those keys. 
\end{itemize}
The host reads these inputs from the DBMS and passes them as column arrays across a C interface that \qsem{} generates. A fixed function $\alpha$ maps machine values to model values, since the host passes fixed-width integers, scaled decimals, and byte arrays, while the specification is stated over mathematical integers, rationals, and byte strings. Because ghost arguments are erased, the interface cannot check the contracts. It assumes them on entry, and this is the only place they enter the proof. The declarations therefore determine what the DBMS and the host interface are trusted to deliver, which is why the interface is validated through testing before it is frozen (\Cref{fig:agent-synthesis}).

\sparagraph{Obligations (verified).}
From the input contracts, \qsem{} generates a fixed specification in Verus for each of the three operations (\texttt{build}, \texttt{apply}, \texttt{execute}). Each has two parts: a precondition that the operation may assume and a postcondition that its code must establish. The preconditions are the input contracts above, stated relative to the ghost database, so the database is a ghost argument to every operation. Each postcondition has two parts: a \emph{correctness} part, which constrains any result the operation returns, and a \emph{totality} part, which limits when it may \emph{decline}. We denote \textit{decline} as $\bot$, and it results in the host DBMS answering the query instead. Totality is stated in terms of an envelope, generated from the schema, which bounds each table's size and each value by the limits of the DBMS and the column's SQL type. The envelope therefore contains every database the DBMS can store, and it only limits when an operation may decline. The three specifications are:

\begin{itemize}[leftmargin=*]
\item \texttt{build} requires that its inputs are the scans of $D$. It ensures that any state $s$ it returns has $s.\texttt{db}() = D$, and that it declines only if $D$ lies outside the envelope.
\item \texttt{execute} requires only the type invariant on $s$, so its guarantee holds for any valid state, however that state was produced. It ensures that any output $o$ it returns satisfies $\mathsf{Q}(s.\texttt{db}(), \alpha(p), \alpha(o))$, and that it declines only if $s.\texttt{db}()$ lies outside the envelope or no correct answer fits the result types.
\item \texttt{apply} requires that its deltas take $s.\texttt{db}()$ to some $D'$ and that its rescans are of $D'$ for exactly the keys it requested. It ensures that any state $s'$ it returns has $s'.\texttt{db}() = D'$, and that it declines only outside the envelope or once $s$ has absorbed more changed rows since it was built than a fixed, predefined history bound. The proof needs to show that the requested keys cover the change.
\end{itemize}

\sparagraph{Invariants.}
The agent cannot choose this $s.\texttt{db}()$ freely, because the postconditions of \texttt{build} and \texttt{apply} tie it to the database the host supplied. The agent's freedom lies in $\mathit{Inv}$, which states what the concrete structures mean in terms of the ghost database.
The verified TPC-H Q14 accelerator, for example, divides the range of ship dates into a tree of dyadic blocks. Its invariant states that, for every block, the stored total and promotional revenue equal the corresponding sums over the joined \texttt{lineitem} and \texttt{part} rows of $s.\texttt{db}()$ shipped in that block. \texttt{execute} must then prove from $\mathit{Inv}$ alone that combining the blocks covering the requested date range satisfies $\mathsf{Q}$. A wrong or too weak $\mathit{Inv}$ makes the proof fail rather than an answer wrong, and because \texttt{build} and \texttt{apply} must establish it, it cannot claim more than the code actually maintains.

% \sparagraph{Declining.}
% Totality is proved. Every decline in an accelerator's code carries a proof that its totality clause permits it. The pipeline's requirement that accelerators never decline
% on the frozen tests only re-checks this on the test data. A declined \texttt{execute} is answered by the DBMS, and a declined \texttt{build} or \texttt{apply} leaves the state invalid until the host rebuilds it, which also resets the count of absorbed changes.

\subsection{Verified Stateful Maintenance}
\label{sec:verif-stateful}

The operations above are proved one call at a time. At run time, the
states of all verified accelerators live in a verified store, a
simplified counterpart of the substrate's versioned store
(\Cref{sec:maintenance}). Throughout, we identify a database with its
model under $\alpha$, restricted to the tables and columns an
accelerator reads.

\sparagraph{Store design.}
Three choices make the store's proof tractable. First, a single store holds every verified accelerator, so each commit installs one version that updates all affected accelerators at once. Second, versions are immutable. A writer prepares the next version outside the lock, calling \texttt{apply} or \texttt{build} on the newest version's states, and takes the write lock only to append it.
If the version it started from is no longer the newest, the append discards the prepared version and leaves a single version whose states are all invalid, so readers fall back until the auxiliary states are rebuilt. The store admits at most one pending commit at a time. Third, a reader selects a version for its snapshot under the read lock and then runs \texttt{execute} on it outside the lock. As a result, no accelerator proof has to reason about interleavings.

\sparagraph{Ghost state and lock invariant.}
The store numbers versions $0, 1, 2, \ldots$ in installation order, and its ghost state is a timeline $\tau$, supplied by the host when the store is created, where $\tau(v)$ is the database that version $v$ represents. Read queries, however, arrive with host epochs rather than version numbers, so each version $v$ also records the engine epoch $f_v$ from which it is current. 
% We show later, in Lemma \ref{lem:snap}, that a reader at epoch $E$ is given the newest version current at $E$. 
The store keeps a window of consecutive versions $b, \ldots, n$, ending at the newest version $n$. 
To bound memory, it prunes versions older than a host-supplied floor, and because it removes only the oldest, the window is always a suffix of the installed versions. 
The lock invariant states that $f_b < \cdots < f_n$ and that every valid state in a kept version $v$ represents $\tau(v)$. Since the lock carries the invariant in Verus, it checks that every write preserves it before releasing the lock, so the readers can just assume it without reasoning about concurrent writers.

\sparagraph{Store operations.} The contracts of \Cref{sec:verif-boundary} are stated for an arbitrary
database $D$ and its successor $D'$; the store instantiates them along the timeline, with $D = \tau(v)$ and $D' = \tau(v{+}1)$. Three operations carry the store's guarantees:
\begin{itemize}[leftmargin=*]
    \item \texttt{prepare} packages the next version's states, computed from the newest version $v$. \textit{Precondition:} every valid state represents $\tau(v+1)$. \textit{Postcondition:} the prepared version shares the store's timeline.
    \item \texttt{append} installs a prepared version. \textit{Precondition:} the version shares the store's timeline. \textit{Postcondition:} the lock invariant holds. The version is installed as $v+1$ if $v$ is the newest. 
    \item \texttt{read} selects a version for snapshot epoch $E$.
   \textit{Precondition:} none.
   \textit{Postcondition:} if no pending commit may become visible at or
   before $E$, it returns a \emph{reading} of the newest kept version $v$
   with $f_v \le E$, whose valid states represent $\tau(v)$; otherwise,
   or if no kept version has $f_v \le E$, the reader falls back.
\end{itemize}

We now show what these contracts guarantee, first for writes and then for reads.
\textbf{Verus checks every lemma and theorem in this section}, so we omit their
proofs and sketch only the idea behind \Cref{lem:snap}.

\begin{lemma}[Sequential maintenance]
\label[lemma]{lem:seq}
If every call to \texttt{prepare} meets its precondition, then after any
sequence of commits, rebuilds, and aborts, every valid state $s$ in every
kept version $v$ satisfies $s.\texttt{db}() = \tau(v)$.
\end{lemma}

% \begin{proof}
% This is the store's lock invariant, checked by Verus.
% Creating the store establishes it from \texttt{build}'s postcondition.
% A prepared version carries \texttt{prepare}'s precondition with it, and
% \texttt{append} installs it as $v{+}1$ only if $v$ is still the newest
% version, so its valid states represent $\tau(v{+}1)$; otherwise
% \texttt{append} installs only invalid states.
% To meet \texttt{prepare}'s precondition, \texttt{apply} starts from a
% state of version $v$, which by the invariant represents $\tau(v)$, so the
% per-call obligations chain from version to version.
% \end{proof}

\Cref{lem:seq} holds for every choice between \texttt{build} and
\texttt{apply}, so the cost model that makes this choice
(\Cref{sec:cost-model}) affects only performance and need not be
trusted.
% Also, by the totality clauses, if every $\tau(v)$ lies in the accepted input envelope (\Cref{sec:verif-boundary}), then no call declines (falls back to base SQL).

\begin{lemma}[Snapshot selection]
\label[lemma]{lem:store}
Consider a reader at snapshot epoch $E$, and fix the window at the moment the reader holds the read lock.
If the reader is given version $v$, then:
\begin{enumerate}[label=(\alph*), leftmargin=*]
\item $b \leq v \leq n$, $f_v \leq E$, and $f_u > E$ for every $u$ with
 $v < u \leq n$, so $v$ is the newest installed version with
 $f_v \leq E$;
\item no commit is pending with $\hat{f} \leq E$; and
\item every valid state in $v$ represents $\tau(v)$.
\end{enumerate}
\end{lemma}

% \begin{proof}
% This is the postcondition of \texttt{read}, checked by Verus. A ghost
% log of every installed version's epoch, of which the window is always
% a suffix, extends (a) to all installed versions.
% \end{proof}
   
\sparagraph{Connecting the store to the host DBMS.}
The lemmas above are stated over $\tau$. However, the store cannot actually observe the DBMS, so it cannot check that $\tau$ is the DBMS's history or that a reader's epoch reflects its snapshot. Verus checks the remaining steps against a ghost model of the engine that defines the database $D_E$ a snapshot at epoch $E$ sees, and the store itself enforces that versions are installed in commit order. What the store cannot check forms the host's \textbf{commit-order contract} (which appear in Verus as trusted preconditions): (i) each installed version $v$ carries the epoch $f_v$ at which its
commit becomes visible, and $\tau(v) = D_{f_v}$;
(ii) when a reader at epoch $E$ reads, every commit visible at $E$ is
either installed or marked pending;
(iii) an accelerator whose state is carried over into a commit's version
reads no table the commit changes; and
(iv) the states the host stages are produced by \texttt{build} or
\texttt{apply} from inputs that satisfy their contracts along $\tau$.

\begin{lemma}[Snapshot correspondence]
\label[lemma]{lem:snap}
Under the commit-order contract, if a reader at epoch $E$ is given version $v$, then $\tau(v) = D_E$, where $D_E$ is the database its snapshot sees.
\end{lemma}

\begin{proof}
Checked by Verus.
Intuitively, every commit visible at $E$ is installed or pending when
the reader reads (ii), none is pending (\Cref{lem:store}(b)), so all are
installed at or before $v$, and $\tau(v) = D_{f_v} = D_E$ by (i).
\end{proof}

% The host also passes a pruning floor. Because pruning keeps a suffix of the installed versions, a wrong floor can only make readers fall back. The store makes no liveness claim.

\subsection{Properties Proved}
\label{sec:verif-guarantee}

We now combine the per-operation obligations with the store's guarantees into two end-to-end results. The first is the correctness of every answer returned, and the second is a bound on when an accelerator may decline, or fall back to the host DBMS.

\begin{theorem}[Correctness]
\label{thm:verif}
Assume that every call into a verified accelerator satisfies the input contracts of \Cref{sec:verif-boundary}, and that the host satisfies the commit-order contract. Then for every reader whose snapshot sees $D_E$ and every parameter value $p$, if the accelerator returns an output $o \neq \bot$, then $\mathsf{Q}(D_E, \alpha(p), \alpha(o))$.
\end{theorem}

% \begin{proof}
% Checked by Verus.
% The reader runs \texttt{execute} on the reading it was given, which
% relies only on the reading's type invariant, so the host cannot
% substitute another state.
% If contract item (ii) held when the reader read, the output satisfies
% $\mathsf{Q}(D_E, \alpha(p), \alpha(o))$ by \Cref{lem:store,lem:snap} and
% \texttt{execute}'s obligation.
% \end{proof}

Correctness alone would also hold for an accelerator that always declines. The totality clauses of \Cref{sec:verif-boundary} exclude this:

\begin{theorem}[No spurious declines]
\label{thm:total}
Under the same contracts, for every database the DBMS can store, \texttt{build} never declines, \texttt{execute} declines only when the query has no answer or no answer fits the result types, and \texttt{apply} declines only after exceeding its history bound.
\end{theorem}

% \begin{proof}
% The totality clauses of the three obligations are checked by Verus; every database the DBMS can store lies in the envelope by its definition.
% \end{proof}

\Cref{thm:total} covers the accelerator, not the store.
A reader can still fall back to the host while a commit that may be
visible at its epoch is pending, while its accelerator's state is
invalid (after a stale append, an abort, or a declined \texttt{build} or
\texttt{apply}, until the rebuild completes), or when its version has
been pruned.
These cases are part of \texttt{read}'s checked postcondition, so the
store falls back only in them.

The guarantees trust the specification (\qsem{}'s rules, engine
profiles, and DataFusion), the verifier toolchain, the host boundary
(C interfaces, glue, scans, change capture, output formatting, and
fallback), and the commit-order contract (\Cref{tab:verif-chain}). The accelerator's data structures, query execution, and maintenance logic are all verified.

\section{Evaluation}
\label{sec:eval}
% more crisp questions
% Our evaluation answers four questions:
% \begin{itemize}[leftmargin=*]
%   \item \textbf{Read performance.} How much faster does \thesystem{} answer analytical queries than its host DBMS and prior agentic systems, and how well does it generalize to workloads unlikely to appear in LLM training data? (Section~\ref{sec:read-performance})
%   \item \textbf{Writes.} Do \thesystem{}-generated accelerators retain speedups even under updates, and how does it compare with incremental view maintenance? (Section~\ref{sec:update-performance})
%   \item \textbf{Maintenance policy.} Does the cost-model selector choose the right maintenance action as the write pattern shifts? (Section~\ref{sec:update-performance})
%   % \item \textbf{Verification.} How much performance do verified accelerators give up compared with differentially tested ones, and where does that cost come from? (Section~\ref{sec:verification-cost})
%   \item \textbf{Generated designs.} How does \thesystem{}'s pipelines compare to state-of-the-art harnesses like Claude Code, and what optimizations does the pipeline discover? (Section \ref{sec:ablation}, Section~\ref{sec:generated-examples})
% \end{itemize}
We ask three questions: how much faster \thesystem{} answers queries than its hosts and prior synthesis systems (\Cref{sec:read-performance}), whether that advantage survives writes (\Cref{sec:update-performance}), and what the agent pipeline contributes and costs (\Cref{sec:ablation}, \Cref{sec:generated-examples}).

\subsection{Experimental Setup}
% \begin{itemize}[leftmargin=*]
%   \item Models used
%   \item Machines used
%   \item Datasets/workload description
%   \item Base database systems (DuckDB, maybe PostgreSQL)
%   \item Competitor agentic systems (BespokeOLAP, GenDB)
% \end{itemize}

We evaluate \thesystem{} on two workloads. The first is TPC-H~\cite{tpch} with all 22 queries, at SF=20 and, for the QphH experiment (\Cref{sec:update-performance}), at SF=30 and SF=3000. Each run draws fresh parameter bindings from the TPC-H substitution rules, so an accelerator cannot specialize to a single binding. The second is SEC-EDGAR at SF=3, built by GenDB~\cite{gendb} from real SEC financial-statement filings from 2022-2024. GenDB provides 1,000 SQLSmith-generated queries over it; we use six of them, with zero to three joins, a nested-aggregate filter, and an anti-join, to test whether the pipeline generalizes to a workload unlikely to appear in LLM training data. We describe the write workloads with the experiments that use them.

\begin{figure*}[t]
    \centering
    \includegraphics[width=\textwidth]{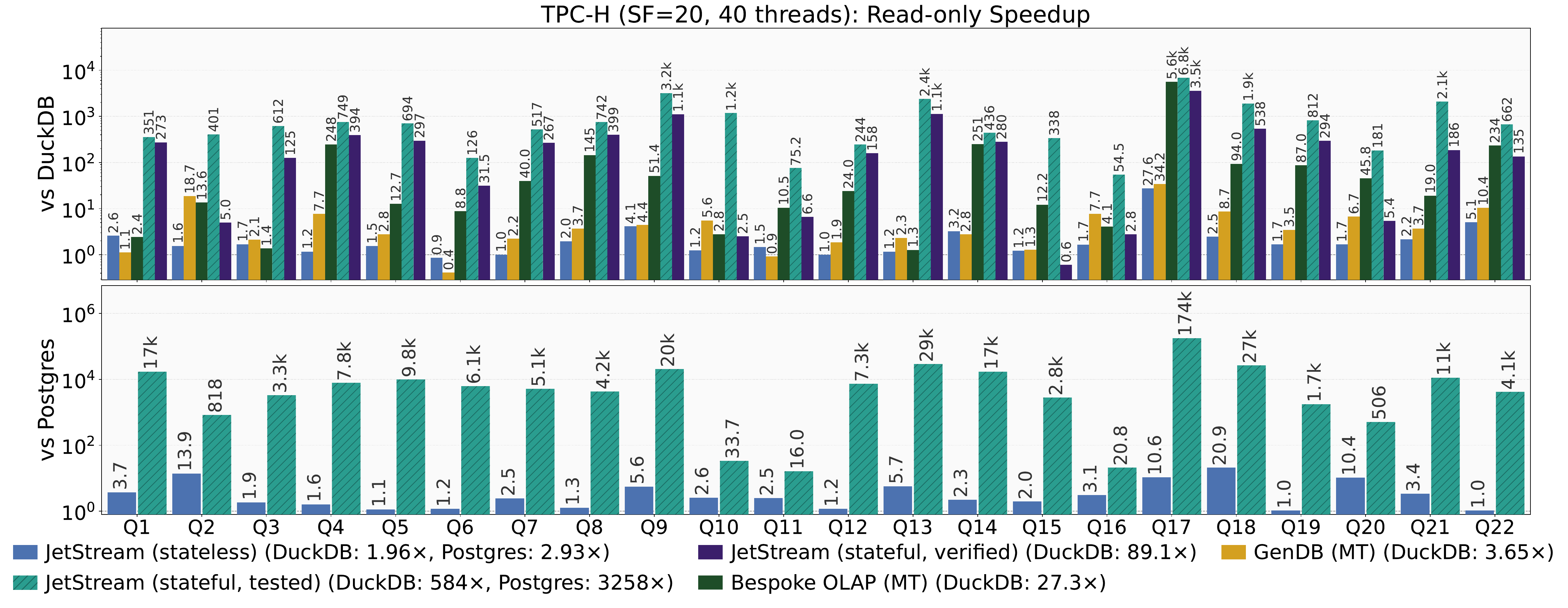}
    \caption{\textbf{\thesystem{} achieves large read-only speedups across TPC-H.}
    Per-query speedup at SF=20 (log scale). All systems use 40 threads except verified \thesystem{}, which runs single-threaded; legend values are geometric means over the 22 queries.
    \emph{Top:} speedup over DuckDB. Stateful \thesystem{} reaches $584\times$ and its verified
    variant $89.1\times$, compared with $27.3\times$ for Bespoke OLAP and $3.65\times$ for
    GenDB. Stateless \thesystem{} reaches $1.96\times$.
    \emph{Bottom:} speedup over PostgreSQL: $3258\times$ for stateful and $2.93\times$ for
    stateless \thesystem{}.}
    \label{fig:tpch-sweep}
\end{figure*}

We evaluate the stateless, tested, and verified variants of
\Cref{sec:system-architecture}, each synthesized for DuckDB v1.5.4 and
PostgreSQL 18.4, where \thesystem{} runs as an extension. We compare
against the hosts themselves; against GenDB~\cite{gendb} and Bespoke
OLAP~\cite{bespoke-olap}, which re-encode the data into a layout they
choose for the workload; and against materialized views maintained by
pg\_ivm 1.16 on PostgreSQL~17, the newest version it supports.

Experiments run on a dual-socket server with two Intel Xeon Gold 6230 CPUs (40 physical cores, 80 hardware threads), 376\,GiB of memory, and a 1\,TB NVMe SSD. Unless otherwise noted, every system uses 40 threads: DuckDB, PostgreSQL, GenDB, Bespoke OLAP, and \thesystem{}'s stateless and tested accelerators. Verified accelerators run single-threaded but are still compared against the 40-thread hosts. 

\subsection{Steady-State Read Performance}
\label{sec:read-performance}

We first measure read performance without writes, the setting GenDB and Bespoke OLAP were designed for (Figure~\ref{fig:tpch-sweep}).

\sparagraph{Stateless accelerators.}
Without auxiliary state, \thesystem{} can only optimize execution over the host's own storage. Stateless accelerators reach geomean speedups $1.96\times$ over DuckDB and $2.93\times$ over PostgreSQL. Both GenDB ($3.65\times$) and Bespoke OLAP ($27.3\times$) both do better, because they re-encode the data into custom format. \textbf{An accelerator working over the engine's own storage has little room left to improve, and most of our gains are from low-level execution optimizations, such as using SIMD instructions for filtering and using bitmaps for expensive operations like joins.}
% \todo{fix english}

% \todo{explain RF1, RF2}
\sparagraph{Stateful accelerators.}
Auxiliary state changes this. Stateful \thesystem{} reaches $584\times$ over DuckDB and $3258\times$ over PostgreSQL, and beats both prior systems on every query: GenDB by $7.1\times$ (Q16) to about $1{,}000\times$ (Q13), and Bespoke OLAP by $1.2\times$ (Q17) to about $1{,}800\times$ (Q13). The gain comes from removing work from the read path. A typical TPC-H query scans and joins much of the database to produce few rows, while a stateful accelerator maintains a parameter-independent summary of that work, keyed by the attributes the parameters range over. The tested Q14 accelerator, for example, keeps promotional and total revenue per ship month and answers a binding by reading one month. What remains at read time is mostly producing the result, so the smallest speedups belong to the queries with the largest outputs: Q16 returns 27,840 rows and Q11 returns 19,000. With little computation left, read latency is dominated by fixed per-query costs, such as dispatch through the host. Read latency changes by under 8\% between one and 40 threads, which is why single-threaded verified accelerators remain competitive. \textbf{Auxiliary state drives \thesystem{}'s advantage: stateful \thesystem{} outperforms both prior agentic system baselines on every TPC-H query.}

\sparagraph{Verified accelerators.}
Verified accelerators reach $89\times$ over DuckDB while running on a single thread against 40-thread DuckDB, and 15 of the 22 queries exceed $100\times$. Six queries (Q2, Q10, Q11, Q15, Q16, Q20) account for most of the gap to the tested accelerators, and verified Q15 is even slower than DuckDB ($0.6\times$). For these queries, the tested accelerator stores an answer-shaped structure, while the verified one keeps a more generic aggregate that is easier to prove correct under maintenance and finishes the query at read time (\Cref{sec:generated-examples} examines this trade-off for Q6). 
% Tested Q15, for instance, stores each month's top suppliers; verified Q15 stores revenue per supplier in 256-day bins and merges across all suppliers on every call (\Cref{sec:generated-examples} examines this trade-off for Q6). 
\textbf{Verified accelerators trail tested ones by $6.5\times$ in geomean, yet outperform Bespoke OLAP by $3.3\times$ and GenDB by $24\times$.}
% On the other 16 queries, the verified accelerators are $2.7\times$ slower in geomean: they run on one thread and read from persistent, reference-counted trees, which add allocation and pointer chasing to the read latency.

\begin{figure}[t]
    \centering
    \includegraphics[width=0.9\columnwidth]{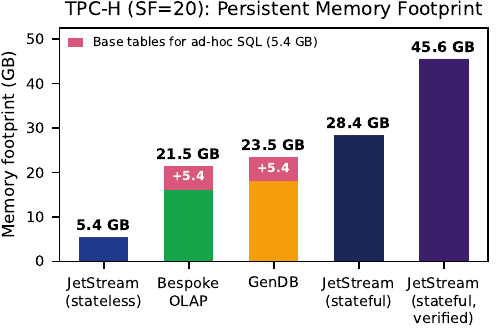}
    \caption{\textbf{Persistent memory footprint on TPC-H at SF=20.}
    Stateful \thesystem{} uses additional memory to maintain query-specific
    auxiliary state.}
    \label{fig:memory-footprint}
    \vspace{-3mm}
\end{figure}
% \todo{plot peak persistent memory}
\begin{figure}[t]
    \centering
    \includegraphics[width=\columnwidth]{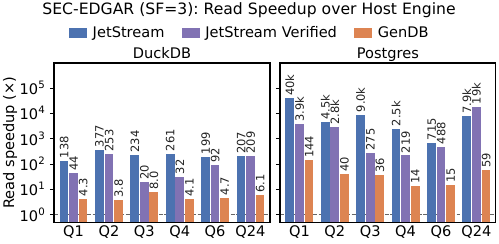}
    \caption{\textbf{\thesystem{} performance on unseen SEC-EDGAR queries} (SF=3, log scale).
    \thesystem{} achieves a $225\times$ ($72\times$ verified) geomean read-only speedup over DuckDB and $5335\times$ ($1350\times$ verified)
    over PostgreSQL, compared with $5.0\times$ and $37\times$ for GenDB.}
    \label{fig:sec-edgar-speedup}
    \vspace{-1.65em}
\end{figure}

\sparagraph{Memory footprint.}
These speedups cost space (Figure~\ref{fig:memory-footprint}). We count what each system must keep resident to answer arbitrary SQL. GenDB and Bespoke OLAP keep the base tables alongside their encoded copies (23.5\,GB and 21.5\,GB). Stateful \thesystem{} needs 28.4\,GB: its state adds 23\,GB over the stateless configuration and raises the geomean from $1.96\times$ to $584\times$. The verified stateful \thesystem{} needs 45.6\,GB. The extra 17 GB comes from favoring representations that are easy to prove correct over compact ones (\Cref{sec:generated-examples}). For example, rather than packing rows into flat arrays, the verified accelerators allocate each row and string separately on the heap. \textbf{For 5-7\,GB more than either prior system, it is $21\times$ faster than Bespoke OLAP and $160\times$ faster than GenDB.} 

\sparagraph{Generalization to SEC-EDGAR.}
\thesystem{} also generalizes well to other workloads like SEC-EDGAR (Figure~\ref{fig:sec-edgar-speedup}): $225\times$ over DuckDB against GenDB's $5.0\times$, and $5335\times$ over PostgreSQL against $37\times$. Verified accelerators, whose specifications \qsem{} generates for all six queries, reach a $72\times$ geomean speedup over DuckDB and $1,350\times$ geomean speedup over PostgreSQL. \textbf{\thesystem{}'s advantage carries over to a real-world workload that LLMs are unlikely to have seen, for both tested and verified accelerators.}

%   \item We need to compare with GenDB + BespokeOLAP, but should consider the
%   right setting. Maybe if the workload is write-heavy enough, we include it in
%   the headline figure.
% \end{itemize}
\subsection{Performance Under Updates}
\label{sec:update-performance}

We next evaluate whether \thesystem{} keeps its read advantage while the database changes underneath it. We run four experiments, each with its own write workload: a continuous TPC-H refresh stream, a comparison with incremental view maintenance, the official TPC-H composite metric, and a stress test of the maintenance policy.

\sparagraph{Continuous refresh stream.}
For each TPC-H query at SF=20, eight concurrent readers repeatedly run the query with parameters drawn from eight random sets, and each query's accelerator runs in its own process. One writer alternates TPC-H's two refresh functions~\cite{tpch}. After each refresh, the writer also updates rows in the tables RF1 and RF2 never touch, so queries that read only those tables, such as Q2, Q11, and Q16, also see the tables that they depend on, change. We vary the write rate from 2.5K to 30K modified rows/s. At the highest rate, about 1\% of the database changes every minute.

Read latency alone would overstate \thesystem{}'s benefit, since its writes must also maintain auxiliary state. We therefore charge each system for both. If DuckDB completes N reads during the run, the \textbf{workload speedup }is the time DuckDB needs for those N reads plus its commits, divided by the time \thesystem{} needs for the same N reads plus its commits and related maintenance: 
\[
  \text{speedup} = \frac{N \cdot \bar{r}_{\text{DuckDB}} + C_{\text{DuckDB}}}{N \cdot \bar{r}_{\thesystem{}} + C_{\thesystem{}}},
\]
where $\bar{r}$ is a system's mean read latency under the concurrent write stream and $C$ is its total commit time for the shared write stream.

\begin{figure}[t]
\centering
\includegraphics[width=\columnwidth]{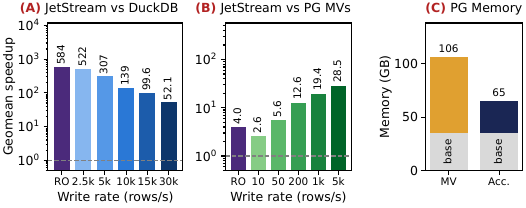}
\caption{\textbf{\thesystem{} keeps large speedups under continuous TPC-H refreshes} (SF=20, eight readers, one refresh writer).
\textbf{(A)}~Geomean workload speedup over DuckDB.
\textbf{(B)}~Geomean workload speedup over pg\_ivm materialized views in PostgreSQL.
\textbf{(C)}~PostgreSQL memory, including about 35\,GB of base tables.}
\label{fig:tpch-updates}
\vspace{-2em}
\end{figure}

\thesystem{} achieves a $522 \times$ geomean workload speedup at 2500 rows/s, $139\times$ at 10,000 rows/s, and $52\times$ at 30,000 rows/s (Figure \ref{fig:tpch-updates}A). Readers see consistent state while maintenance runs, so the decline comes entirely from maintenance, whose cost grows linearly with the write rate. Maintenance is also selective. About half of the writes for a given query need none, such as point updates to tables it does not read, and \thesystem{} skips them. \textbf{Even when 1\% of the database changes every minute, \thesystem{} remains $52\times$ faster than DuckDB.}

\sparagraph{Incremental view maintenance.}
IVM is the standard way to keep precomputed state fresh, so we compare against materialized views maintained by pg\_ivm in PostgreSQL. The views group by the attributes the parameters filter on, so, like \thesystem{}, they answer any parameter binding. 
% Q21 and Q22 use \texttt{NOT EXISTS}, which pg\_ivm does not support, so they are plain views refreshed inside every write transaction that touches their inputs. 
With Claude Code~\cite{anthropic2025claudecode}, we build the best MV setup for our baseline. This includes indexes on group keys, view definitions that pg\_ivm can maintain with index lookups, and no JIT for maintenance statements. Writes follow the refresh stream above, at 10 to 5K modified rows/s.

pg\_ivm maintains every view synchronously inside the writing transaction, and a single TPC-H-sized refresh commit takes 4 to 17 minutes. The views can therefore only sustain at most 241 modified rows/s, against 5.8K for our \thesystem{} extension. Beyond that rate, the views fall behind, so Figure~\ref{fig:tpch-updates}B reports workload speedup, which charges both systems for the write stream. The views' reads, however, still run against a lagging snapshot, which favors them. Even so, \thesystem{} tested accelerators are $4.0\times$ faster read-only, and its advantage grows to $28.5\times$ at 5K rows/s. It also uses less memory: 65\,GB against 106\,GB, both including about 35\,GB of base tables (Figure~\ref{fig:tpch-updates}C). \textbf{\thesystem{} sustains $24\times$ the write rate of pg\_ivm while running up to $28.5\times$ faster and using 39\% less memory.}

\sparagraph{TPC-H composite metric.}
We also measure the official TPC-H metric, QphH@Size \cite{tpch}. The power test runs RF1, the 22 queries in sequence, and RF2 in a single stream, and scores their geometric-mean time. The throughput test runs $S$ concurrent query streams alongside a refresh stream that executes $S$ RF1/RF2 pairs, and scores the time until the last stream finishes. We use $S=4$ at SF=30 and $S=8$ at SF=3000, as the official TPC-H specification states. QphH@Size is the geometric mean of the two scores. For \thesystem{}, every refresh includes maintenance of all affected accelerators. 

\begin{table}[t]
\caption{Improvement of \thesystem{} (tested, stateful accelerators) over DuckDB on the TPC-H composite metric.}
\label{tbl:qphh}
\vspace{-0.5em}
\small
\begin{tabular}{@{}lrr@{}}
\toprule
Metric & SF=30 & SF=3000 \\
\midrule
Power@Size      & $96.2\times$  & $79 \times$\\
Throughput@Size & $1.82\times$ & $1.63 \times$\\
QphH@Size       & $13.1\times$ & $11.4 \times$\\
\bottomrule
\end{tabular}
\vspace{-1em}
\end{table}

At SF=30, \thesystem{} improves Power@Size by $96.2\times$ and QphH@Size by $13.1\times$ over DuckDB, and at SF=3000, a database 100 times larger, it still improves them by $79\times$ and $11.4\times$ (\Cref{tbl:qphh}). Throughput@Size improves by only $1.6$-$1.8\times$, because the two tests are bounded by different work. Accelerators shrink the 22 query timings that dominate the power test to milliseconds. The throughput test ends only when the refresh stream finishes, however, so once queries become cheap, refreshes and their maintenance set the full duration. \textbf{On the standard TPC-H metric, \thesystem{} improves QphH@Size by $11$-$13\times$, and the improvement holds from SF=30 to SF=3000.}

\begin{figure}
    \centering
    \includegraphics[width=\columnwidth]{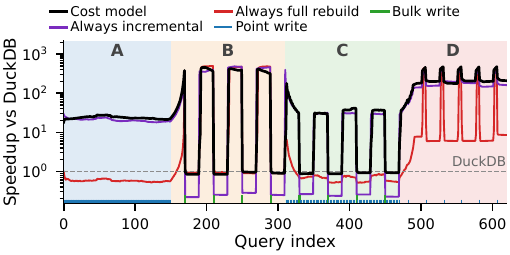}
    \caption{\textbf{\thesystem{} adapts its maintenance strategy to the write pattern} (TPC-H Q1, SF=20). Phases A-D are described in the text.}
    \label{fig:cost-model-phases}
    \vspace{-1.5em}
\end{figure}

\sparagraph{Maintenance policy.}
Point and bulk writes favor different maintenance strategies. To test the cost-model selector (\Cref{sec:cost-model}), we run Q1 620 times and interleave writes in four phases (Figure~\ref{fig:cost-model-phases}): Phase~A has a single-row point write before every query, Phase~B a bulk write every 40 queries, Phase~C the bulk writes of Phase~B plus a point write every other query, and Phase~D a read-mostly mix with a point write every 25 queries. We compare the selector with two fixed policies, always incremental and always full rebuild. Under point writes, the selector applies changes incrementally, while always rebuilding reconstructs the whole state for every small update. Rebuilding is $41\times$ slower than the selector in Phase~A ($0.57\times$ versus $24\times$ over DuckDB) and $30\times$ slower in Phase~D ($7.3\times$ versus $219\times$). Under bulk writes the trade-off reverses: applying a bulk batch's deltas takes 17--20\,s, while a full rebuild takes about 4\,s. In Phase~B, the selector therefore switches to rebuilding and matches the always-rebuild policy ($1.51\times$ versus $1.56\times$). Phase~C needs both choices, and there the selector is $2.4\times$ faster than always rebuilding and $3.6\times$ faster than always applying deltas. \textbf{The selector matches or beats the better fixed policy in every phase, and over the whole workload it is $2.7\times$ faster than DuckDB, while neither fixed policy beats DuckDB ($0.75\times$ and $0.96\times$).}

\subsection{Agent Pipeline Ablation}
\label{sec:ablation}

The pipeline assigns a model to each role. For our tested pipeline, DeepSeek-V4.1-Flash ~\cite{deepseekv41} implements and optimizes accelerators, GPT-6.1-Sol~\cite{gpt61sol} serves as the advisor, and Qwen3.8-Max~\cite{qwen38} performs audits and reviews. For verified accelerators, MiMo-V2.6-Pro~\cite{mimov26} writes the accelerators and their proofs. It is cheap enough for the many iterations proofs require but often stalls, and the GPT-6.1-Sol advisor that redirects it accounts for most of the cost. Synthesizing the DuckDB and PostgreSQL backends cost \$28.14 and \$7.73, respectively, or \$35.87 in total. The 22 TPC-H accelerators cost \$15.22 in the stateless variant, \$41.63 in the tested variant, and \$334.65 (about \$15 each) in the verified variant, and the six SEC-EDGAR tested accelerators cost \$10.92.

\begin{figure}[t]
    \centering
    \includegraphics[width=\columnwidth]{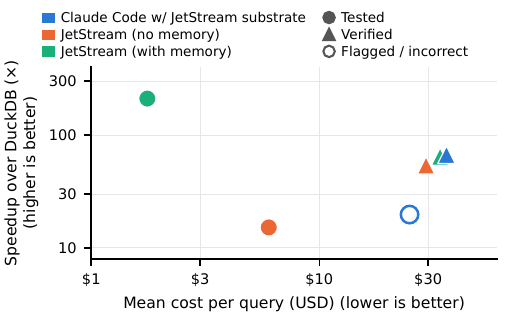}
    \caption{\textbf{Speedup over DuckDB against mean synthesis cost per query (for subset of queries)}, comparing \thesystem{} with and without learned memory and a single Claude Code agent (Opus 5.5) with the same \thesystem{} substrate.}
    \label{fig:ablation}
    \vspace{-1em}
\end{figure}

To see where this efficiency comes from, we synthesize a subset of TPC-H query accelerators (Q1, Q4, Q6, Q7, Q11, Q12, Q16) with three generators: \thesystem{}, \thesystem{} without its learned memory (\Cref{sec:evidence-optimization}) or
proof library (\Cref{sec:verification}), and a single Claude Code agent (Opus 5.5 ~\cite{opus55}) given the same substrate, tests, and acceptance criteria (Figure~\ref{fig:ablation}). For tested accelerators, \thesystem{} with learned memory reaches a $211\times$ geomean speedup over DuckDB at \$1.76 per query. Without memory, it reaches only $15.2\times$ at \$6.00, and the single agent reaches $19.8\times$ at \$24.87. The Claude Code single agent's run is flagged: it hardcoded accelerators for the parameter bindings in the exposed tests and failed on the hidden full set, which \thesystem{} prevents by having the orchestrator own the tests and the gates that run them. For verified accelerators, the three generators end up close together, at $54$-$67\times$ and \$29-36 per query on this harder subset, because proof effort dominates both cost and design. The single Claude Code agent is marginally fastest ($67.1\times$) at the highest cost (\$36.12). \textbf{For tested accelerators, the staged pipeline with learned memory is $11\times$ faster than a single agent at $14\times$ lower cost; for verified accelerators, the proof itself, not the pipeline, sets the cost.}

\subsection{Generated Accelerator Examples}
\label{sec:generated-examples}

\begin{figure}[t]
    \centering
    \includegraphics[width=\columnwidth]{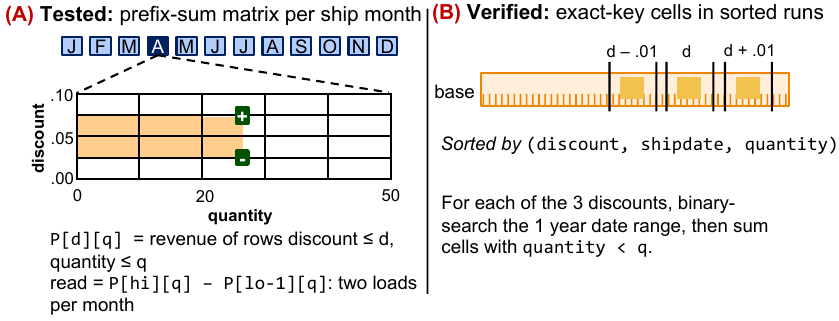}
    \caption{\textbf{Auxiliary state for TPC-H Q6.} (a) Tested: a prefix-sum matrix per ship month. (b) Verified: exact-key cells in sorted runs.}
    \label{fig:q6-aux-state}
    \vspace{-2em}
\end{figure}

We examine one stateless and one stateful generated accelerator (both tested and verified variants), on TPC-H \cite{tpch}. 

% To summarize the database briefly: its largest table, \texttt{lineitem} (about 6M rows per scale factor), holds one row per item in an order, with its quantity, price, discount, and ship date; \texttt{orders}, \texttt{customer}, \texttt{part}, \texttt{supplier}, and \texttt{partsupp} describe the orders, buyers, products, and suppliers, and \texttt{nation} and \texttt{region} are small lookup tables.

\sparagraph{Stateless TPC-H Q17.}
Q17 sums the price of line items whose quantity is below 20\% of their part's average, for parts of one brand and container. Only about 0.1\% of \texttt{lineitem} rows qualify, and the kernel exploits this: it records the qualifying part keys in a bitmap, then makes one pass over DuckDB's column arrays, probing eight keys at a time with AVX-512 gathers and computing the averages and the sum over the hits alone. It is $27.6\times$ faster than DuckDB, the largest stateless speedup.

\sparagraph{Stateful TPC-H Q6.}
Q6 sums revenue over one year of ship dates, a discount range $d \pm 0.01$, and a quantity bound $q$. The tested accelerator keeps a prefix-sum matrix of revenue over discount and quantity for each ship month, so a read costs two loads per month (Figure~\ref{fig:q6-aux-state}A). It assumes TPC-H's value domains (discounts 0--0.10, whole quantities up to 50) and declines any month with a row outside them. The verified accelerator stores exact (discount, shipdate, quantity) cells in sorted runs (Figure~\ref{fig:q6-aux-state}B). Since \texttt{l\_discount} has two decimal places, $d \pm 0.01$ covers exactly three values, so each read binary-searches each run three times. 

The two states are nearly the same size, but the tested accelerator is $126\times$ faster than multi-threaded DuckDB against $32\times$ for the single-threaded verified one. The dense matrix is cheap to read but hard to prove correct: one changed row updates every cumulative cell above it, and the proof must maintain a two-dimensional range-sum invariant. Exact-key runs need only a per-key sum invariant, at the cost of searches at read time. \textbf{Proof obligations push verified accelerators toward generic structures with simple invariants, which leave more work for reads; this pattern explains most of the verification gap in \Cref{sec:read-performance}.}

\section{Related Work}\label{sec:related-work}

\sparagraph{Synthesized accelerators.}
GenDB~\cite{gendb} and Bespoke OLAP~\cite{bespoke-olap} synthesize
workload-specific engines with LLMs; \Cref{sec:intro} and
\Cref{tbl:qual-comparison} contrast them with \thesystem{}.
GPT-DB~\cite{gptdb-trummer23} generates query code interactively, and
Castor~\cite{castor-feser20} specializes data layouts for static templated
workloads. Unlike all four, \thesystem{} keeps one copy of the data in the
host's native storage, runs ad-hoc queries unchanged, and maintains its
auxiliary state transactionally under writes.
%
% \thesystem{} also leverages auxiliary state to accelerate specific queries, but
% it generates incremental maintenance logic for this state, allowing it to
% preserve most of its accelerated performance while delivering fresh results in
% the presence of writes.

% GPT-DB~\cite{gptdb-trummer23} also generates code to execute a query, but does
% so interactively.
% %
% Castor~\cite{castor-feser20} is a classical approach to specializing the data
% layout to a static templated query workload; it does not support writes to the
% underlying data.

\sparagraph{Database extensibility.}
Extensions~\cite{dbextsurvey-kim25, pgextensions, duckdb-extensions},
UDFs~\cite{udfprism-arch25, udo-sichert22}, and declarative
sub-operators~\cite{dsos-jungmair23} let users add accelerators to a DBMS,
but leave them to design, implement, maintain, and port each one.
Tailwind~\cite{tailwind-yu26} integrates user-written accelerator families
non-invasively and selects instances for a workload. \thesystem{} instead
has agents identify and implement accelerators and their maintenance, and
can prove them correct.

\sparagraph{Learned specialization.}
There is a long line of prior work that uses machine learning techniques to
specialize a database system or data infrastructure to specific query
workloads~\cite{copyright-sudhir21, tsunami-ding20, flood-nathan20,
qdtree-yang20, sagedb-ding22, brad-yu24, neo-marcus19, bao-marcus21,
snarf-vaidya22, lsort-kristo20, plbf-vaidya21, treeline-yu22,
selfdrivingoperation-pavlo21}. At a high level, these techniques define a parameterized space of possible
system configurations and leverage learned models or analytical models to
predict performance across this space, allowing the system parameters to be
optimized for a given workload.
\thesystem{} is similar in spirit in that it specializes query execution to a
given workload, but it does not require a manually defined parameterized
configuration space; instead, LLM agents perform the workload analysis, code
generation, and optimization.

\sparagraph{Query compilation.}
Compiling engines such as Umbra~\cite{umbra-neumann20},
HyPer~\cite{hyper-kemper11}, and
Redshift~\cite{redshiftindustry-armenatzoglou22} generate code from a fixed
recipe per operator~\cite{querycompilation-neumann11}. \thesystem{}
generates code specialized to one query and its data, including auxiliary
state.

\sparagraph{Incremental view maintenance.}
\thesystem{} incrementally maintains any auxiliary state it chooses to generate,
which is similar to incremental view maintenance~\cite{ivm-zhou07, dbsp-budiu23,
dbtoaster-ahmad09, noria-gjengset18, naiad-murray13}. However, the key
difference is that incremental view maintenance comprises algorithms for
maintaining state computed from a SQL expression. In contrast, \thesystem{}
generates incremental maintenance code for arbitrary auxiliary state. In other
words, the auxiliary state \thesystem{} uses is not necessarily always a
database view.

\sparagraph{Verified systems and proof synthesis.}
Closest to \thesystem{}'s verified accelerators is Inductive Deductive Synthesis (IDS)~\cite{ids-agarwal26}, which uses agents to jointly synthesize implementations and proofs of distributed key-value stores. IDS starts from a hand-written specification and identifies writing it as the main bottleneck; \thesystem{} instead derives the specification mechanically from SQL, using our translator system \qsem{}. Other LLM-based work generates proofs for given code~\cite{autoverus-yang25, baldur-first23} or verified code for single functions~\cite{clover-sun24}. Verified data systems~\cite{verifiedrdbms-malecha10, veribetrkv-hance20, ironfleet-hawblitzel15} are proven once by hand, whereas \thesystem{} generates a verified accelerator per query. Work on SQL semantics and equivalence~\cite{spja, sqlcert-benzaken19, cosette-chu17} reasons about queries themselves rather than their implementations.
\vspace{-0.5em}

\iffalse
\begin{itemize}[leftmargin=*]
  \item \textbf{Synthesizing accelerators.} BespokeOLAP, GenDB, GPT-DB, Castor
  \item \textbf{Database extensibility.} UDFs, UDOs, declarative sub-operators (DSOs), Tailwind (user provided accelerators)
  \item \textbf{Learned system specialization.} SageDB, Bao/Neo, Tsunami/Flood/Qdtree, TreeLine, BRAD
  \item \textbf{Query compilation.} Neumann, HyPer/Umbra
  \item \textbf{Incremental view maintenance.} DBSP, materialize, noria, DBToaster, other classical work
  % \item \textbf{Coding agents?}
\end{itemize}
\fi

\section{Conclusion}

We presented \thesystem{}, a framework for generating query-specific database accelerators that extend an existing DBMS, rather than replace it. By building engine- and query-specific code with agentic pipelines around a fixed engine-neutral substrate, \thesystem{} allows accelerators to run over native storage and maintain their auxiliary state under
writes, yielding a $584\times$ geomean speedup over 40-threaded DuckDB. \thesystem{} can also generate accelerators with a machine-checked proof against a trusted specification derived from the SQL while retaining an $89\times$ read speedup over DuckDB. More broadly, our results suggest that LLM-based specialization need not replace the engine, or make assumptions about static data. A production DBMS can remain responsible for storage, transactions, recovery, and ad-hoc SQL, while generated components can specialize only the expensive parts of a workload.

\bibliographystyle{ACM-Reference-Format}
\bibliography{codegen}

\end{document}